\documentclass[manuscript,screen,nonacm]{acmart}

\usepackage{bm}
\usepackage{xcolor}
\usepackage{xspace}
\usepackage{enumitem}
\usepackage{tcolorbox}
\usepackage[noend]{algpseudocode}

\tcbuselibrary{breakable}

\acmConference[PODC '21]{ACM Symposium on Principles of Distributed Computing}{July 2021}{Virtual Event}
\acmYear{2021}
\copyrightyear{2021}

\setcitestyle{numbers,sort&compress}

\iffalse

\newcommand{\highlight}[2][red]{{\color{#1} #2}}

\newcommand{\marginnote}[2][teal]{\marginpar{\footnotesize\color{#1} #2}}

\else

\newcommand{\highlight}[2][black]{{\color{black} #2}}

\newcommand{\marginnote}[2][white]{\marginpar{\phantom{\color{white} #2}}}

\fi

\begin{document}

\newtheorem{claim}{Claim}[theorem] 
\newtheorem*{claim*}{Claim} 


\title{Tight Trade-off in Contention Resolution without Collision Detection}

\author{Haimin Chen}
\orcid{}
\affiliation{
	\department{State Key Laboratory for Novel Software Technology}
	\institution{Nanjing University}
	\country{China}}
\email{haimin.chen@smail.nju.edu.cn}

\author{Yonggang Jiang}
\orcid{}
\affiliation{
	\department{State Key Laboratory for Novel Software Technology}
	\institution{Nanjing University}
	\country{China}}
\email{yonggangjiang@smail.nju.edu.cn}

\author{Chaodong Zheng}
\orcid{}
\affiliation{
	\department{State Key Laboratory for Novel Software Technology}
	\institution{Nanjing University}
	\country{China}}
\email{chaodong@nju.edu.cn}

\begin{abstract}
In this paper, we consider contention resolution on a multiple-access communication channel. In this problem, a set of nodes arrive over time, each with a message it intends to send. In each time slot, each node may attempt to broadcast its message or remain idle. If a single node broadcasts in a slot, the message is received by all nodes; otherwise, if multiple nodes broadcast simultaneously, a collision occurs and none succeeds. If collision detection is available, nodes can differentiate collision and silence (i.e., no nodes broadcast). Performance of contention resolution algorithms is often measured by throughput---the number of successful transmissions within a period of time; whereas robustness is often measured by jamming resistance---a jammed slot always generates a collision. Previous work has shown, with collision detection, optimal constant throughput can be attained, even if a constant fraction of all slots are jammed. The situation when collision detection is not available, however, remains unclear.

In a recent breakthrough paper [Bender et al., STOC '20], a crucial case is resolved: constant throughput is possible without collision detection, but only if there is no jamming. Nonetheless, the exact trade-off between the best possible throughput and the severity of jamming remains unknown. In this paper, we address this open question. Specifically, for any level of jamming ranging from none to constant fraction, we prove an upper bound on the best possible throughput, along with an algorithm attaining that bound. An immediate and interesting implication of our result is, when constant fraction of all slots are jammed, which is the worst-case scenario, there still exists an algorithm achieving a decent throughput: $\Theta(t/\log{t})$ messages could be successfully transmitted within $t$ slots.
\end{abstract}

\maketitle


\section{Introduction}\label{sec-intro}

\emph{Contention resolution} is a classical problem in parallel and distributed computing. In this problem, there are multiple players trying to access a shared resource; and the problem is solved once each player has successfully accessed the resource (at least) once. This problem is complicated by the requirement that accesses must be \emph{mutual exclusive}: if two or more players try to utilize the shared resource simultaneously, none would succeed. Contention resolution saw various applications in computer science, some prominent examples include congestion control in computer networking (e.g., Ethernet and IEEE 802.11 wireless networks \cite{kurose17,metcalfe76}), concurrency control in database management systems and operation systems (e.g., locking \cite{ramakrishnan02,tanenbaum14}).\marginnote{More examples and refs?}

In this paper, we consider the following more concrete setting, which is a standard model used in many papers studying the problem (see, e.g., \cite{bender18,bender20}). The shared resource is a \emph{multiple-access communication channel}, and time is divided into discrete and synchronized \emph{slots}. Each player (also called a \emph{node}) joins the system at the beginning of some slot with a single \emph{message} (also called a \emph{packet}) it intends to send.  We assume node arrival times are controlled by an adaptive adversary. In each slot, each node can either try to broadcast its message on the channel or remain idle. In case a single node broadcasts, the transmission succeeds and all participating nodes receive the unique message; otherwise, if multiple nodes broadcast simultaneously in a slot, a \emph{collision} occurs and all transmission attempts in that slot fail. Upon collision, the exact channel feedback depends on the availability of a \emph{collision detection} mechanism. Specifically, with collision detection, nodes can tell whether a wasted slot (i.e., a slot without a successful message transmission) is due to silence (i.e., no node broadcasts in this slot) or collision (i.e., multiple nodes broadcast in this slot). By contrast, if collision detection is not available, nodes cannot differentiate silent slots and colliding slots.

Achieving a high \emph{throughput} is the primary goal of contention resolution algorithms. Though exact definition of ``throughput'' differ in various papers, intuitively it measures how many messages an algorithm can successfully transmit within a period of time.

Transmission attempts could fail even without collision, as the shared communication channel could suffer from hardware/software errors (e.g., a shared printer may crash due to hardware/software bugs) or unintentional/intentional external interference (e.g., a wireless link may be affected by electromagnetic noise). In the context of contention resolution, such failures are often modeled by \emph{jamming} (see, e.g., \cite{awerbuch08,richa10,chang18,bender18}). Formally, if a slot is jammed, then a collision occurs on the channel in that slot, regardless of the actual number of broadcasting nodes.

An important question in studying various contention resolution algorithms is to understand their robustness against jamming. More fundamentally, in solving contention resolution, what is the inherent trade-off between throughput---measuring performance, and jamming-resistance---measuring fault-tolerance? As it turns out, the availability of collision detection makes a huge difference. More specifically, in case collision detection is available, even if a constant fraction of all slots could be jammed (which is the asymptotic worst-case scenario), constant throughput (which is the asymptotic optimal throughput) can be attained~\cite{bender18,chang18}. By contrast, in the absence of collision detection, even without jamming, whether optimal throughput is possible for the vanilla contention resolution problem remains unknown for a long time. Last year, in a breakthrough paper~\cite{bender20}, Bender et al.\ showed that such algorithm does exist. Moreover, they have also proved an impossibility result implying constant throughput is impossible when constant fraction of slots could be jammed. However, an intriguing question remains open: in solving contention resolution, what is the exact trade-off between throughput and jamming-resistance, when collision detection is not available?

\smallskip\noindent\textbf{Main contribution.} In this paper, we answer the above question by providing a tight and complete characterization of the trade-off between the best possible throughput and the severity of jamming. In particular, for any level of jamming ranging from none to constant fraction, we have an upper bound on the best possible throughput, along with a corresponding algorithm attaining that bound. Together with previous work, this marks the complete understanding of the throughput versus jamming-resistance trade-off for the contention resolution problem.

In the reminder of this section, we will discuss some additional model details, then proceed to a more careful introduction of our main results. We will conclude this section with a brief survey on related work.

\smallskip\noindent\textbf{Additional model details.} We assume each node has a single message to send, and a node will leave the system immediately once its message has been successfully transmitted. We assume node arrival times and jamming are both controlled by an adaptive adversary. We often call this adversary Eve, and her adaptivity is reflected by the assumption that, in each slot, she could use past channel feedback to determine whether to jam the slot, and whether to inject (one or more) new nodes. Notice, the multiple-access channel provide identical feedback to the nodes and the adversary, meaning Eve also does not posses the ability of collision detection.

\smallskip\noindent\textbf{Statement of results.} Throughout this paper, we call a slot \emph{active} if there is at least one player in the system in that slot. Following classical definition of throughput is adopted from \cite{bender20}: let $n_t$ be the number of players arriving the system in the first $t$ slots, and let $a_t$ be the number of active slots among the first $t$ slots, the $throughput$ at slot $t$ is defined to be $n_t/a_t$. An algorithm achieves a certain throughput $\lambda$ if for all slots the throughput is lower bounded by $\lambda$. That is to say, for instance, if an algorithm achieves a constant throughput, then the number of active slots is at most some constant factor larger than the number of player arrivals.

In this paper, we extend the above definition to better reflect the influence of jamming, see following.

\begin{definition}[$(f,g)$-throughput]\label{def:f-g-throughput}
Let $\mathcal{A}$ be the algorithm each node runs after arriving. Denote the number of newly arrived players and the number of jammed slots in the first $t$ slots as $n_t$ and $d_t$, respectively. Let $f,g:\mathbb{R}^+\rightarrow\mathbb{R}^+$ be two functions. We say $\mathcal{A}$ achieves \emph{$(f,g)$-throughput} if for any integer $t\ge 1$ and any adaptive adversary strategy, the number of active slots in the first $t$ slots is at most $n_t\cdot f(t) + d_t\cdot g(t)$, with high probability in $n_t$.\footnote{An event happens with high probability (w.h.p.) in some parameter $\lambda$ if it happens with probability at least $1-1/\lambda^c$, for some constant $c\geq 1$.}
\end{definition}

Notice that the above definitions of throughput (both the classical one and $(f,g)$-throughput) do not claim a bound on the number successful transmissions within a time interval. Nonetheless, as Bender et al.~\cite{bender20} have pointed out, they could imply such results. To see this, consider an algorithm that achieves $(f,g)$-throughput and an interval of length $t$. If the adversary does not jam too many slots or inject too many nodes in the sense that $n_t\cdot f(t)+d_t\cdot g(t)<\lambda t$, then one of the most recent $\lambda t$ slots is inactive, implying any node arriving before $(1-\lambda)t$ has succeeded. We capture this relationship between $(f,g)$-throughput and the number of successes more precisely in Corollary~\ref{cor:throughput-and-successes} in Section \ref{sec-alg-analysis}.

As mentioned earlier, our main result concerns with the trade-off between the best possible throughput and the severity of jamming. Interestingly, this comes down to a trade-off between $f$ and $g$. To see this intuitively, notice that successfully sending $n_t$ messages in $t$ slots requires $n_t\cdot f(t) + d_t\cdot g(t)< t$, which implies $(n_t/t)\cdot f(t) + (d_t/t)\cdot g(t)\leq 1$. This suggests, if an algorithm wants to send $n_t$ messages within $t$ slots, then $n_t/t$---which roughly corresponds to throughput---would be at most $1/f$, and $d_t/t$---which corresponds to severity of jamming---would be at most $1/g$. This hints the throughput versus jamming-resistance trade-off is in fact a trade-off between $f$ and $g$.

The following two theorems state the exact relationship between functions $f$ and $g$:

\begin{theorem}[Algorithmic Result]\label{thm:upper-bound}
\highlight{For any function $g$ such that $\log^2(g)$ is sub-logarithmic, there exist a function $f$ where $f(x)\in\Theta\left(\frac{\log x}{\log^2g(x)}\right)$ and an algorithm achieving $(f,g)$-throughput.}
\end{theorem}

\begin{theorem}[Impossibility Result]\label{thm:lower-bound}
\highlight{For any functions $f$ and $g$ such that $f$ and $\log^2(g)$ are both sub-logarithmic, if $f(x)\in o\left(\frac{\log x}{\log^2g(x)}\right)$, then there does not exist any algorithm achieving $(f,g)$-throughput.}
\end{theorem}

\textit{Remark 1.} Throughout this paper, we say a function $f:\mathbb{R}^+\rightarrow\mathbb{R}^+$ is \emph{sub-logarithmic} if: (1) $f(x)\in O(\log x)$ and is non-decreasing; (2) for some large constant $y_0$, there exists some constant $x_0$ such that $f(x)\geq y_0$ when $x\geq x_0$; (3) for any constant $c>0$, there exists some constant $c_0$ such that for any $n\in\mathbb{N}^+$, $|f(cn)-f(n)|\leq c_0$; (4) for any constant $c>0$, $f(x^c)\in\Theta(f(x))$. We define ``sub-logarithmic'' this way so as to avoid artificial pathological functions.

\textit{Remark 2.} \highlight{In Theorem \ref{thm:upper-bound}, when $\log{g(x)}\in\Theta(\sqrt{\log{x}})$, $f$ becomes a constant function and our algorithm could achieve constant throughput, which is the best possible throughput.}

To illustrate the application of the above results, consider two interesting cases. Suppose $g$ is a constant function, meaning Eve can jam some constant fraction of all $t$ slots, then the best ratio of the number of player arrivals over $t$---which roughly corresponds to throughput---is $\Theta(1/\log{t})$. Moreover, we can devise an algorithm that attains this throughput. (This means in the absence of collision detection, even with constant fraction of jamming, we could send $\Theta(t/\log{t})$ messages in $t$ slots, achieving a decent---though sub-constant---throughput.) On the other hand, if we would like $f$ to be some constant function, meaning sending $\Theta(t)$ messages within $t$ slots, then the maximum number of slots Eve can jam must be bounded by $2^{\Theta(\sqrt{\log{t}})}$. Again, we can devise an algorithm that tolerates such a jamming adversary.

\smallskip\noindent\textbf{Related work.}
One simple and standard algorithm to resolve contention is \emph{binary exponential backoff}. In its classical implementation (e.g., in Ethernet~\cite{metcalfe76}), each participating node waits a random time interval before trying to broadcast its message; if a transmission attempt failed, the node waits another randomly chosen time interval before retrying, and the expected length of the waiting interval doubles after each failure.

Unfortunately, simple exponential backoff cannot provide optimal throughput, for both statistical arrival patterns~\cite{aldous87,hastad87} and batch/adversarial arrival patterns~\cite{bender05}. In view of this, numerous variations of the binary exponential backoff scheme have been proposed and analyzed (such as polynomial backoff and saw-tooth backoff), again for both statistical arrival patterns (e.g., \cite{capetanakis79,raghavan95}) and batch/adversarial arrival patterns (e.g., \cite{bender05,chlebus12}). This paper considers adversarial arrival pattern, and the proposed algorithm utilizes two exponential backoff variants as key subroutines.

Simple backoff algorithms usually do not depend on the availability of collision detection: nodes decrease sending probabilities (i.e., backoff) whenever an empty slot is observed. However, this behavior is not always correct: an empty slot could also mean no node tries to broadcast; in such case, nodes should be more aggressive and increase their sending probabilities (i.e., backon). Therefore, with collision detection, more clever backoff-backon algorithms can be devised (e.g., \cite{awerbuch08,richa10,bender18,chang18}). Such algorithms are especially helpful if external interference is present: often optimal performance can be attained in spite of jamming. However, when collision detection is not available, the exact impact of jamming on solving contention resolution remains unclear. Our paper addresses this open question.

Beside throughput, another important metric when evaluating contention resolution algorithms is the number of channel accesses a node has to make before successfully sending its message (some authors call this the energy complexity). Many existing algorithms (e.g., \cite{bender18,bender20}) have $O(\textrm{poly-log}(n))$ energy complexity, assuming there are $n$ nodes in the system. Nonetheless, somewhat surprisingly, Bender et al.~\cite{bender16} show that $O(\log(\log^* n))$ channel accesses per node is enough for resolving contention.\marginnote{Discuss our algorithm's energy complexity?}

It is worth noting, if we only care about the first success (instead of requiring each node to succeed once), then the problem essentially degrades to leader election---a classical symmetry breaking task. In a seminal work by Willard~\cite{willard86}, a tight bound of $\Theta(\log\log{n})$ slots is proved, assuming $n$ nodes are activated simultaneously. In case nodes are injected dynamically, the problem is known as the ``wake-up problem''~\cite{chlebus05} or the ``synchronization problem''~\cite{dolev09}.

Lastly, we note that contention resolution is an extensively studied problem, and many interesting results are not covered here. (E.g., there are papers focusing on deterministic algorithms~\cite{anantharamu19,marco19}, and papers considering performance metric other than throughput~\cite{agrawal20}.) Interested readers are encouraged to find dedicated survey papers for more details.

\smallskip\noindent\textbf{Paper outline.} In Section \ref{sec-alg-description}, we will first give an overview of the algorithm that achieves optimal throughput for any given level of jamming, including some key design ideas; and then provide a complete description of the algorithm. In Section \ref{sec-alg-analysis}, we will analyze the proposed algorithm and prove our algorithmic result---Theorem \ref{thm:upper-bound}. We will also prove a corollary connecting $(f,g)$-throughput and number of successful transmissions. Finally, in Section \ref{sec-lower-bound}, we will prove two impossibility results. The first one is Theorem \ref{thm:lower-bound}, while the second one demonstrates a certain type of exponential backoff cannot achieve optimal throughput, justifying some decisions we made during the algorithm design process.

\section{The Algorithm}\label{sec-alg-description}

Our algorithm has same high-level framework as \cite{bender20}, but with key adjustments made specifically for achieving the best possible throughput against jamming.

\smallskip\noindent\textbf{Algorithm framework.} It is known that binary exponential backoff cannot provide constant throughput, even if $n$ nodes are activated simultaneously. Nevertheless, in such ``batch'' scenario, the first $\Theta(n)$ slots of the process could achieve a constant throughput. To see this, consider the following implementation of binary exponential backoff: each node broadcasts with probability $1/i$ in slot $i$. For this algorithm, around slot index $n$, if $\Theta(n)$ nodes have already succeeded then we are done. Otherwise, at least $\Theta(n)$ nodes remain, and the sum of their broadcasting probability---often called the \emph{contention} of a slot---is $\Theta(1)$. As the contention of a slot corresponds to the expected number of broadcasting nodes in that slot, a constant contention means a successful transmission will occur with constant probability. Therefore, starting from slot $n$, after another $\Theta(n)$ slots, there is a good chance that at least $\Theta(n)$ successes will occur.

However, such high throughput cannot be maintained in later portion of binary exponential backoff. Hence, we need a mechanism to stop the process once $\Theta(n)$ slots are executed, and then restart. If nodes can access \emph{two} independent channels, then the following method would work. On one channel, called the ``data channel'', nodes execute the above standard exponential backoff algorithm. On the other channel, called the ``control channel'', nodes execute a modified backoff algorithm. The goal of the modified backoff algorithm is to let the first success of the control channel occur in slot $\Theta(n)$, so that backoff on the data channel can stop at the right time. As it turns out, the modified backoff algorithm is pretty simple: each node broadcasts with probability $(\log{i})/i$ in slot $i$.

In the vanilla contention resolution problem, nodes are injected dynamically over time, thus another mechanism is required to ``synchronize'' nodes, so that they can start a backoff process on the data channel in a batch manner. Again, with two independent channels, a simple solution exists. Specifically, a newly arrived node first runs exponential backoff on the control channel, until a success occurs on the control channel. (A new node cannot simply listen and wait for a success, as it might be the only node in the system.) At that point, all nodes in the system are synchronized and can start an efficient new batch on the data channel.

At this point, the only remaining issue is that the model only provides one channel. If nodes have access to a global clock, then an easy solution would be: (1) groups odd slots together and call it ``odd channel'', then assign odd channel to be control channel; and (2) groups even slots together and call it ``even channel'', then assign even channel to be data channel. Unfortunately, such global clock is also not available. Therefore, we need yet another mechanism to allow nodes to make consensus on the role of slots. We defer the discussion of this mechanism to algorithm description.

\smallskip\noindent\textbf{Achieving jamming resistance.} In the above algorithm framework, two types of backoff algorithms are used in two different settings, with different purposes: (1) truncated exponential backoff in batch setting, with the goal of achieving $\Theta(n)$ successes in $\Theta(n)$ slots, assuming $n$ nodes start simultaneously; (2) standard exponential backoff in dynamic setting, with the goal of achieving a single success efficiently. It turns out that the truncated exponential backoff process is extremely robust against jamming. In particular, among the first $\Theta(n)$ slots, even if a constant fraction is jammed, the procedure could still guarantee $\Theta(n)$ successes, and halt at the correct time. On the other hand, however, standard backoff performs poorly against jamming in the dynamic arrival setting. Specifically, in the extreme case in which a single node executes standard exponential backoff, if the adversary jams early slots, then the node's sending probability quickly decays to sub-optimal values, resulting it taking too much time to succeed.

A natural fix to the above issue is to decrease nodes' sending probabilities slower. But to what extent? After all, if nodes' sending probabilities remain high for too long, then in the other extreme case in which Eve injects a lot of nodes within a short period of time, contention among nodes themselves would prevent quick first success. This dilemma is exactly what we exploit in proving the impossibility results (see Section \ref{sec-lower-bound}), and it also hints the optimal sending probabilities nodes should use when running backoff style algorithms in dynamic arrival setting.

\subsection{Algorithm Description}\label{subsec-alg-description}

We first introduce two (parameterized) subroutines: the \emph{\textsc{backoff}} subroutine and the \emph{\textsc{batch}} subroutine. As the name suggests, the \textsc{backoff} subroutine aims to achieve quick first success in dynamic arrival setting, whereas the \textsc{batch} subroutine aims to achieve good throughput in batch setting. Both subroutines are variants of the standard exponential backoff algorithm.
(Careful readers might wonder why we use two different variants, the reason being: the \textsc{batch} subroutine simplifies algorithm presentation and analysis; while the \textsc{backoff} subroutine is necessary for achieving optimal throughput, see Theorem \ref{thm:non-adaptive-lower-bound} in Section \ref{sec-lower-bound} for details.)
Below are the definitions of the two subroutines.

\vspace{-3ex}\begin{center}\begin{tcolorbox}
[
width=0.9\textwidth,
breakable, 
standard jigsaw, opacityback=0, 
opacityframe=0, 
]
\begin{itemize}
	\item[\underline{$h$-\textsc{backoff}}:] Let $h:\mathbb{N}^+\rightarrow\mathbb{N}^+$ be a function. We say a node runs $h$-\textsc{backoff} starting from slot $l$, if for any $k\in\mathbb{N}$, in slot interval $\mathcal{I}_k=[l-1+2^k,l-1+2^{k+1})$, the node sends its message in slots $\{l_i\mid1\le i\le h(|\mathcal{I}_k|)\}$ where each $l_i$ is drawn uniformly at random (with replacement) from $\mathcal{I}_k$. We call $\mathcal{I}_k$ as the $k$-th stage of $h$-\textsc{backoff}.
	\item[\underline{$h$-\textsc{batch}}:] Let $h:\mathbb{N}^+\rightarrow\mathbb{R}^+$ be a function. We say a node runs $h$-\textsc{batch} starting from slot $l$, if for any $k\in\mathbb{N}^+$, the node sends its message with probability $\min\{1,h(k)\}$ in slot $l-1+k$.
\end{itemize}
\end{tcolorbox}\end{center}\vspace{-3ex}

Our algorithm requires a function $g$ as an input parameter, where $\log{g(x)}=O(\sqrt{\log{x}})$. This function signifies the level of jamming the algorithm cam tolerate.
(Recall Definition \ref{def:f-g-throughput} and discussions below it.)
For instance, if Eve can jam some constant fraction of all slots, then $g(x)$ should be a constant function; whereas if Eve can jam $1/\log{t}$ fraction of all slots over a time period of $t$, then $g(x)=\log{x}$.
\marginnote{Is known $g(x)$ necessary for achieving optimal throughput?}
Given $g(x)$, define function $f(x)=\frac{ac_2\log x}{\log^2 (g(x)/a)}$; further define function $h^{ctrl}(x)=\frac{c_3\log{x}}{x}$ and function $h^{data}(x)=\frac{1}{x}$. Here, $a,c_2,c_3$ are constant to be determined in later analysis.

Conceptually, our algorithm uses two channels: the \emph{odd channel} which contains odd slots, and the \emph{even channel} which contains even slots. We use $\alpha$ to represent one of the two channels, and use $\bar{\alpha}$ to represent the other channel.

We are now ready to state the algorithm for a newly injected node $u$. It contains three phases:

\vspace{-3ex}\begin{center}\begin{tcolorbox}
[
width=0.9\textwidth,
breakable, 
standard jigsaw, opacityback=0, 
opacityframe=0, 
]
\begin{itemize}
	\item[\textbf{Phase 1:}] Suppose $u$ is injected at the beginning of slot $l_0$. Node $u$ will run $(\frac{1}{a}f)$-\textsc{backoff} on the channel determined by the parity of $l_0$,\footnotemark\ until hearing a success on either one of the two channels. Node $u$ will then start Phase 2.
	\item[\textbf{Phase 2:}] Suppose the first success $u$ observed during Phase 1 occurred in slot $l_1$, and $l_1$ is in channel $\alpha$. Node $u$ will run $(\frac{1}{a}f)$-\textsc{backoff} on channel $\bar{\alpha}$ starting from slot $l_1+1$, until hearing a success on channel $\bar{\alpha}$ in some slot $l_2$. Node $u$ will then start Phase 3, with $l_3$ set to $l_2$.
	\item[\textbf{Phase 3:}] Starting from slot $l_3+1$ node $u$ will run $h^{ctrl}$-\textsc{batch} on the channel determined by the parity of $l_3+1$, and starting from slot $l_3+2$ node $u$ will run $h^{data}$-\textsc{batch} on the channel determined by the parity of $l_3+2$. Let $\alpha$ be the channel on which $u$ runs $h^{ctrl}$-\textsc{batch}. One execution of Phase 3 ends when $u$ hears a successful transmission on channel $\alpha$ in some slot $l'_3$. By then, $u$ sets $l_3$ to $l'_3$, and restarts Phase 3.
\end{itemize}
\end{tcolorbox}\end{center}\vspace{-3ex}
\footnotetext{Node $u$ does not need to know whether $l_0$ is in the odd channel or the even channel.}

Recall the algorithm framework introduced at the beginning of this section. Phase 1 allows a newly joined node $u$ and the existing nodes to reach agreement on the role (data and control) of the two channels (odd channel and even channel). Specifically, the newly joined node $u$ treats the channel on which the success occurred as the data channel. (Notice, since it might be the case that there are only Phase 1 nodes in the system, nodes in Phase 1 cannot just passively wait for successes. Instead, they run \textsc{backoff} to create the first success efficiently.) Once all nodes have reached agreement on the role of the two channels, in Phase 2, node $u$ runs $(\frac{1}{a}f)$-\textsc{backoff} on the control channel and waits for a success to occur on the control channel. Once such a success occurs, node $u$ and the other nodes executing Phase 2 or Phase 3 are synchronized, and can (re)start Phase 3, which contains an execution of the ``truncated exponential backoff''. Lastly, one important detail worth noting is, whenever a node (re)starts Phase 3, it swaps its data channel and control channel.

Finally, we note that a node halts once its message has been successfully transmitted, as specified by the model.

\section{Algorithm Analysis}\label{sec-alg-analysis}

\subsection{Analysis of \textsc{backoff}}

In this subsection, we will prove two key lemmas demonstrating the effectiveness and robustness of the \textsc{backoff} subroutine. They are used extensively in later analysis.

We begin by stating a concentration inequality that will be used in proving these two lemmas.

\begin{theorem}
[McDiarmid's Ineqality~\cite{dubhashi12}]
\label{thm:mcdiarmid}
Suppose $f(x_1,x_2,\cdots,x_n)$ is a function satisfying: for any $i\in[n]$, it holds $|f(\vec{x})-f(\vec{x}')|<c$, where $\vec{x}=(x_1,x_2,\cdots,x_i,\cdots,x_n)$ and $\vec{x}'=(x_1,x_2,\cdots,x'_i,\cdots,x_n)$ only differ in the $i$-th coordinate. Suppose $\{X_i\}_{i\in[n]}$ are $n$ independent random variables. Then for any $\delta>0$:
\[\Pr\left[\left|f({X}_1,{X}_2,...,{X}_n)-\mathbb{E}\left[f({X}_1,{X}_2,...,{X}_n)\right]\right|\ge\delta\right]\le 2\exp{\left(-\frac{2\delta^2}{c^2n}\right)}\]
\end{theorem}

The first key technical lemma concerns within the scenario where \textsc{batch} and \textsc{backoff} are being executed concurrently. That is, in the system, a set of (synchronized) nodes are running \textsc{batch}, while some other (un-synchronized) nodes are running \textsc{backoff}. In such case, so long as the adversary does not inject too many new nodes or jam too many slots, at least one success will occur sufficiently fast. One point worth noting is, conditioned on the state at the beginning of the considered slot interval, prior to the first success, the power of an adaptive adversary and an oblivious adversary are identical (as all channel feedback is silence, an adaptive adversary has nothing to adapt to). Since our applications of the two lemmas only concern with first success, here we only consider an oblivious adversary.

\begin{lemma}\label{lem: batch with additional nodes}
Let $t$ and $c'$ be sufficiently large integers, $f$ be any sub-logarithmic function satisfying $f(x)=O(\log{x})$. Consider slot interval $\mathcal{I}=[L,R]$ where $|\mathcal{I}|=t$, assume the following conditions hold:
\begin{enumerate}
	\item For each slot $i\in[L,R]$, there are $q_i$ (synchronized) nodes running some instance of \textsc{batch}, and each node's sending probability is $p_i$. Moreover, it holds that $\frac{\log t}{c'}\ge q_i\cdot p_i\ge\frac{c'\log t}{t}$ and $p_i\le\frac{1}{2}$.
	\item During slot interval $[1,R]$, aside from the nodes described in (1), an oblivious adversary injects at most $\frac{t}{100f(t)}$ additional new nodes and jams at most $\frac{t}{100}$ slots.
	\item Each node injected by the adversary runs $f$-\textsc{backoff} after joining the system.
\end{enumerate}
Then there exists at least one success slot in $\mathcal{I}$, with high probability in $t^2$.
\end{lemma}

\begin{proof}
Let $mid=(L+R)/2$. Call the nodes injected by the oblivious adversary as \emph{additional} nodes. Let $x_i$ be a random variable denoting the number of additional nodes that send in slot $i$. We consider two complement cases, depending on the value of $\sum_{L\leq i< mid}\mathbb{E}[x_i]$.

\textit{Case 1.} Suppose $\sum_{L\leq i< mid}\mathbb{E}[x_i]\leq t^{0.7}$. Call a slot \emph{occupied} if at least one additional node sends in this slot, or if this slot is jammed.
Note that $X=\sum_{L\le i<mid}{x_i}$ depends on at most $t\log t$ independent random variables: there are at most $\frac{t}{100f(t)}$ additional nodes, each will broadcast in at most $f(t)\log{t}$ slots during $\mathcal{I}$; the indices of the slots these nodes broadcast are the random variables that determine $X$.
It is easy to see, changing the value of each such random variable affects the value of $X$ by at most one. Since $\sum_{L\leq i< mid}\mathbb{E}[x_i]\leq t^{0.7}$, according to Theorem~\ref{thm:mcdiarmid}, with high probability in $t^2$ we have $\sum_{L\le i<mid}{x_i}\le t^{0.8}$. Thus, for sufficiently large $t$, during $[L,mid)$, the number of occupied slots is at most $t/100+t^{0.8}\le t/4$. In each non-occupied slot, since there are $q_i$ nodes sending messages each with probability $p_i$ and $\frac{\log t}{c'}\ge q\cdot p_i\ge\frac{c'\log t}{t}$, it is easy to verify a success occurs with probability $\Omega(\min\{\frac{c'\log t}{t},t^{-2/c'}\})=\Omega(\frac{c'\log t}{t})$ when $c'\geq 2$. Since there are at least $t/4$ non-occupied slots during $[L,mid)$, by a Chernoff bound, we know for sufficiently large $c'$, a success will occur in some non-occupied slot, with high probability in $t^2$.

\textit{Case 2.} Suppose $\sum_{L\le i< mid}\mathbb{E}[x_i]> t^{0.7}$. We now prove $\mathbb{E}[x_i]>t^{-0.4}$ for any $i\ge mid$. For each additional node $u$, denote $a^u_i$ as the probability that $u$ sends in slot $i$. Recall that $f$-backoff will send at most $f(2^j)$ times in the $j$-th stage and the $j$-th stage has length $2^j$. If $u$ arrives before slot $L-t$, the stages that intersect with $[L,R]$ all have length at least $t$. Thus, there are at most two stages intersecting $[L,R]$, and $a^u_i$ in each such stage differ by some constant factor. This implies $a^u_k$ is at least $\frac{1}{t^{1.1}}\sum_{L\le i<mid}a^u_i$ for $k\in[mid,R]$. If, however, $u$ joins after slot $L-t$, we know when $i\in[mid,R]$ it must be the case $a^u_i=\Omega(1/t)$, since the stages intersecting $[mid,R]$ has length $O(t)$. Moreover, $\sum_{L\le i\le mid}a^u_i$ is at most $O(\log^2t)$ since there are at most $\log t$ stages in $[L,mid]$ and each such stage contributes at most $f(t)=O(\log t)$ times of sending. Thus, $a^u_k\ge\frac{1}{t^{1.1}}\sum_{L\le i\le mid}a^u_i$ for $k\in[mid,R]$. As this point, we conclude: for any additional node $u$ and any $k\in[mid,R]$, $a^u_k\ge\frac{1}{t^{1.1}}\sum_{L\le i\le mid}a^u_i$. Since $\mathbb{E}[x_i]=\sum_u{a^u_i}$, we have $\mathbb{E}[x_i]> t^{-0.4}$ for any $i\in[mid,R]$.

Let $x'_i$ be the \emph{total} (i.e., including nodes running \textsc{batch} as well as \textsc{backoff}) number of nodes that send in slot $i$. Since each additional node contributes at most $f(t)\log t$ sending slots during $[L,R]$, we know $\mathbb{E}[\sum_{mid\le i\le R}{x'_i}]\le \frac{t}{100f(t)}\cdot f(t)\log{t}+\frac{t}{2}\cdot\frac{\log t}{c'}=\frac{t\log{t}}{100}+\frac{t\log{t}}{2c'}$. For sufficiently large $c'$, this means there are at most $\frac{t}{4}$ slots in $[mid,R]$ with $\mathbb{E}[x'_i]\ge \frac{\log t}{8}$, which further implies there are at least $\frac{t}{4}$ slots in $[mid,R]$ with $t^{-0.4}\le\mathbb{E}[x'_i]\le\frac{\log t}{8}$. Call a slot $i$ \emph{good} if: (1) $t^{-0.4}\le\mathbb{E}[x'_i]\le\frac{\log t}{8}$; (2) no nodes are injected in $i$; and (3) the adversary does not jam $i$. We know there are at least $\frac{t}{8}$ good slots in $[mid,R]$. Moreover, it is easy to verify, a good slot has probability at least $t^{-0.4}$ to generate a success. Now, let $Y$ be a random variable denoting the number of successes in slots $[mid,R]$. We know $\mathbb{E}[Y]$ is at least $\frac{t}{8}\cdot t^{-0.4}=\frac{t^{0.6}}{8}$. Note that $Y$ is a function of at most $t\log t$ independent random variables each affecting the value of $Y$ by at most one. Thus, according to Theorem \ref{thm:mcdiarmid}, with high probability in $t^2$ we have at least one success in $[mid,R]$.
\end{proof}

The following second lemma focus on the scenario in which no instance of \textsc{batch} is running, or \textsc{batch} ends in the middle of the considered interval. (Specifically, in the lemma statement, slot $k$ marks the end of the \textsc{batch} instance, and $k=0$ means initially there is no instance of \textsc{batch} running). Once again, so long as the adversary does not inject too many new nodes or jam too many channels, successful transmission would soon occur.

\begin{lemma}\label{lem: batch without additional nodes}
Let $t,c',c'_1$ be sufficiently large integers. Let $g$ be any function such that $\log g(x)$ is sub-logarithmic and $\log g(x)=O(\sqrt{\log x})$. Consider slot interval $[1,t]$. There exists a sufficiently large constant $c'_2$ and a function $f(x)=\frac{c'_2\log x}{\log^2g(x)}$ such that, if the following conditions are satisfied:
\begin{enumerate}
	\item There exists an integer $k\in[0,t]$ such that for each slot $i\in[1,k]$, there are $q_i$ (synchronized) nodes running some instance of \textsc{batch}, and each node's sending probability is $p_i$. Moreover, for each $i\in[k/2,k]$, it holds $\frac{\log k'}{c'}\ge q_i\cdot p_i\ge \frac{c'\log k'}{k'}$ and $p_i\leq\frac{1}{2}$, where $k'=\frac{k}{2}$.
	\item Aside from the nodes described in (1), an oblivious adversary injects at most $\frac{t}{c'_1f(t)}$ additional new nodes and jam at most$\frac{t}{c'_1g(t)}$ slots in $[1,t]$.
	\item The adversary injects at least one additional node in the first $\max\{k,1\}$ slots.
	\item Each node injected by the adversary runs $f$-\textsc{backoff} after joining the system.
\end{enumerate}
Then there exists at least one success in the first $t$ slots, with high probability in $t$.
\end{lemma}

\begin{proof}
Let $f(x)=\frac{c'_2\log x}{\log^2g(x)}$ for some sufficiently large $c'_2$ to be specified latter. First consider the situation in which the adversary injects at least $t^{0.6}$ nodes in the first $t/2$ slots. Let random variable $x'_i$ be the number of nodes that send in slot $i$. Since each additional node arrived in the first $t/2$ slots contributes at least $1/t$ to $\mathbb{E}[x'_i]$ when $i>t/2$ and $c'_2$ is sufficiently large, we have the lower bound $\mathbb{E}[x_i]\ge t^{-0.4}$ for any $i>t/2$. On the other hand, since each additional node contributes at most $f(t)\log t$ sending slots in $[1,t]$, and since for existing nodes executing \textsc{batch} (if some \textsc{batch} is running) $q_i\cdot p_i$ is at most $\frac{\log t}{c'}$ when $i>t/2$, we have the upper bound $\mathbb{E}[\sum_{t/2< i\leq t}{x'_i}]\le\frac{t}{c'_1f(t)}\cdot f(t)\log{t} + \frac{t}{2}\cdot \frac{\log t}{c'}$. Then, apply same analysis as in Case 2 of the proof of Lemma~\ref{lem: batch with additional nodes}, there is a success in $[t/2,t]$ with high probability in $t$.

In the reminder of the proof, we consider the situation in which the adversary injects at most $t^{0.6}$ nodes in the first $t/2$ slots. Specifically, we consider three cases depending on the value $k$.

\textit{Case 1:} Suppose $k\ge t/2$. For sufficiently large $t,c',c'_1,c'_2$, all the conditions in Lemma~\ref{lem: batch with additional nodes} are satisfied if we set $L=k/2$ and $R=k$. Thus, apply Lemma~\ref{lem: batch with additional nodes} and we know there is a success in $[L,R]$ with high probability in $(k/2)^2>t$.

\textit{Case 2:} Suppose $t/2\ge k\ge t/g(t)$. Recall that $g(t)$ is at most $2^{O(\sqrt{\log t})}$. Also recall that at most $t^{0.6}$ nodes are injected before slot $t/2$, and in Lemma~\ref{lem: batch with additional nodes} we only require the fraction of jammed slots to be bounded by $1/100$. As a result, for sufficient large $c'_1$, all the conditions in Lemma~\ref{lem: batch with additional nodes} are satisfied if we set $L=k/2$ and $R=k$. Thus, there is at least one success in $[L,R]$ with high probability in $\left(k'\right)^2$. Since $k$ is at least $\frac{t}{g(t)}$ and $g(t)$ is at most $2^{\sqrt{\log n}}$, $\left(k'\right)^2$ is at least $t$ for sufficiently large $t$.

\textit{Case 3.} Suppose $k\le t/g(t)$. Suppose $u$ is a node arriving in slot $t_u$ where $t_u\le k$.
(By lemma assumption, such node $u$ exists.)
We will prove that there is at least one success for $u$ in the first $t/2$ slots. We call a slot \emph{occupied} if additional nodes other than $u$ send in this slot or this slot is jammed by the adversary. Since each additional node can send in at most $f(t)\log t$ slots in $[1,t]$ and there are at most $\frac{t}{c'_1g(t)}\le\frac{t}{g(t)}$ jammed slots, there are at most $t^{0.6}\cdot f(t)\log{t}+\frac{t}{g(t)}\le\frac{2t}{g(t)}$ occupied slots for sufficiently large $t$. (Recall that $\log g(t)=O(\sqrt{\log t})$).

Define slot interval $\mathcal{P}_j=[t_u+2^j-1,t_u+2^{j+1}-1)$. Recall that $u$ will choose $f(|\mathcal{P}_j|)$ slots uniformly at random (with replacement) to send in interval $\mathcal{P}_j$. Suppose there are $a_j$ occupied slots in interval $\mathcal{P}_j$. Define $\ell$ as the largest $j$ such that $\mathcal{P}_j$ is entirely contained within interval $[1,t/2]$. Recall that $\log g$ is sub-logarithmic, which means $g$ is at least $2^{10}$ for sufficiently large $t$, hence $k\le\frac{t}{g(t)}\le\frac{t}{4}$. As a result, $\ell\in\Theta(\log t)$. Define $\ell'$ as the smallest $j$ such that $\frac{2t}{g(t)}\le 2^j$. Clearly $2^{\ell'}>k$ and $\ell\ge\ell'$, meaning only occupied slots can jam the success of $u$ for \textsc{backoff} stages after $\ell'$. Since there are at most $\frac{2t}{g(t)}$ occupied slots, the probability that there is no success slot for node $u$ in $\mathcal{P}_{\ell'},\cdots,\mathcal{P}_{\ell}$ is at most:
$$\prod_{\ell'\le i\le \ell}\left(\frac{a_i}{2^i}\right)^{f(|\mathcal{P}_i|)}
\leq\prod_{\ell'\le i\le\ell}\left(\frac{2t}{2^ig(t)}\right)^{f\left(2^{\ell'}\right)}
\leq\left(\frac{10}{g(t)}\right)^{\frac{1}{2}\cdot f\left(2^{\ell'}\right)\cdot\log\left(g(t)\right)}$$

Since $\ell'$ is the smallest $j$ such that $\frac{2t}{g(t)}\le 2^j$ and $\log g(t)=O(\sqrt{\log t})$, $2^{\ell'}>t^{0.5}$. Since $f(t^d)=\Theta\left(f(t)\right)$ for any $d>0$, $f(2^{\ell'})=\Theta\left(f(t)\right)$. Recall $f(t)=\frac{c'_2\log t}{\log^2 g(t)}$, for sufficiently large $c'_2$, the above probability is at most $1/t$.
\end{proof}

\subsection{Additional Notations}

In this subsection, we present some additional notations so as to simplify the presentation of later analysis.

We begin by introducing three special types of slots.

\marginnote{Definition of transition slot needs double check, can it be beginning/ending slot?}
\vspace{-3ex}\begin{center}\begin{tcolorbox}
[
width=0.85\textwidth,
breakable, 
standard jigsaw, opacityback=0, 
opacityframe=0, 
]
\begin{itemize}
	\item[\underline{\emph{Beginning slot}}:] Consider a slot $t$, if there exists some active node in slot $t$ and all active nodes in the system in slot $t$ were injected at the beginning of slot $t$, then slot $t$ is a {beginning slot}.
	\item[\underline{\emph{Ending slot}}:] Consider a slot $t$, if there is only one active node in the system at the beginning of slot $t$, and it succeeds in slot $t$, then slot $t$ is an {ending slot}.
	\item[\underline{\emph{Transition slot}}:] Consider a slot $t$ containing a success and denote the last beginning slot before $t$ (or $t$ itself if $t$ is a beginning slot) as $b$. The success slot $t$ is called a {transition slot} if either: (1) there is no success within slots $[b,t-1]$; or (2) $t-r_t$ is odd and there is no success within slots $\{s\in[r_t+1,t-1]:(t-s)\textrm{ is a positive even integer}\}$, where $r_t$ is the last transition slot before $t$.
\end{itemize}
\end{tcolorbox}\end{center}\vspace{-3ex}

Intuitively, transition slots are the slots in which active nodes' states change: from Phase 1 to Phase 2, or from Phase 2 to Phase 3, or restart Phase 3.

We then utilize above special slots to define \underline{\emph{complete intervals}}. Specifically, each of the following three kinds of time intervals is a complete interval:

\vspace{-3ex}\begin{center}\begin{tcolorbox}
[
width=\textwidth,
breakable, 
standard jigsaw, opacityback=0, 
opacityframe=0, 
]
\begin{itemize}
	\item Interval $[b,r]$ from a beginning slot $b$ to the first transition slot $r$ after $b$.
	\item Interval $[r+1,r']$ from slot $r+1$ to the first transition slot $r'$ after $r$, where slot $r$ is a transition slot.
	\item Interval $[r+1,e]$ from slot $r+1$ to the first ending slot $e$ after $r$, where: (1) $r$ is a transition slot; (2) $r$ is not an ending slot; and (3) there is no transition slot within $[r+1,e-1]$.
\end{itemize}
\end{tcolorbox}\end{center}\vspace{-3ex}

Intuitively, for any interval, if we mark all the beginning-slots/ending-slots/transition-slots, then these special slots divide the interval into a set of ``segments''. Each such segment that contains active node(s) is a complete interval.

Now, consider an arbitrary interval $\mathcal{I}=[L_{\mathcal{I}},R_{\mathcal{I}}]$. If a node starts Phase 1 or Phase 2 of the main algorithm in $\mathcal{I}$, then the node is a \underline{\emph{new arrival}} of $\mathcal{I}$.
Notice that if some node arrives before slot $L_{\mathcal{I}}$ and begins Phase 2 in slot $L_{\mathcal{I}}$, it is also a new arrival of $\mathcal{I}$ by definition. Therefore, each node is a new arrival of at most two complete intervals.

Lastly, inspired by \cite{bender20}, we define \underline{\emph{truncated length}} for complete intervals.
This facilitates later amortized analysis.
Consider an arbitrary complete interval $\mathcal{I}=[L_{\mathcal{I}},R_{\mathcal{I}}]$. The length of $\mathcal{I}$ is $l_{\mathcal{I}}=R_{\mathcal{I}}-L_{\mathcal{I}}+1$. The truncated length $\bar{l}_{\mathcal{I}}$ is defined in the following way: (1) if the number of new arrivals of $\mathcal{I}$ is at most $al_{\mathcal{I}}/(64cc_3c_1f(l_{\mathcal{I}}))$, and the number of jammed slots during $\mathcal{I}$ is at most $al_{\mathcal{I}}/(64cc_3c_1g(l_{\mathcal{I}}))$, and the number of success during $\mathcal{I}$ is less than $l_{\mathcal{I}}/(32cc_3(t_0+2))$, then $\bar{l}_{\mathcal{I}}=l_\mathcal{I}$; (2) otherwise $\bar{l}_{\mathcal{I}}=0$. Here, $a,c,c_1,c_3,t_0$ are constants to be specified in later analysis; $f$ and $g$ are the functions used by the algorithm, see Subsection~\ref{subsec-alg-description}. (Recall $a$ and $c_3$ are also used in describing the algorithm.)

\subsection{Bounding the Truncated Length of Complete Intervals}

The main goal of this subsection is to show the truncated length of any complete interval is likely to be small. In particular, for any sufficiently large integer $t$, the probability that the truncated length of a complete interval reaching $t$ is at most $1/\textrm{poly}(t)$. Intuitively, this means each complete interval is able to maintain desirable throughput: either the adversary jams a lot of slots or injects a lot of new nodes, or the algorithm generates sufficiently many successes.

We first introduce a technical lemma bounding the sum of a set of dependent random variables. It will be used several times in remaining analysis. We defer its proof to the appendix.

\begin{lemma}\label{lem: special case of sum of dependent random variables}
Let $t_0$ be an arbitrary positive integer. Suppose $X_1,X_2,\cdots,X_n$ are $n$ (potentially dependent) random variables. If $\Pr[X_i=t~|~X_1=x_1, X_2=x_2,\cdots, X_{i-1}=x_{i-1}]\leq t^{-13}$ holds for any integer $t\geq t_0$, any $1\leq i\leq n$, and any values $x_1,x_2,\cdots,x_{i-1}$ of $X_1,X_2,\cdots,X_{i-1}$, then with high probability in $n$, we have $\sum_{i=1}^{n} X_i\leq (t_0+2)n$.
\end{lemma}

We now proceed to bound the truncated length of complete intervals.

\begin{lemma}\label{lem: bounding length of a complete interval}
Consider a complete interval $\mathcal{I}$ that starts at the beginning of slot $L_{\mathcal{I}}$, assume it ends at the end of slot $R_{\mathcal{I}}$. $R_{\mathcal{I}}$ is a random variable. For any integer $t\geq t_0$ where $t_0$ is a sufficiently large constant, we have $\Pr[\bar{l}_{\mathcal{I}}=t]\leq 1/t^{\Omega(1)}$, regardless of the history before slot $L_{\mathcal{I}}$.
\end{lemma}

\begin{proof}
By definition, it is easy to verify that in the beginning slot of a complete interval, each active node in the system will start some instance of \textsc{backoff} or \textsc{batch} from scratch. Thus, without loss of generality, we assume complete interval $\mathcal{I}$ begins at slot one. That is, $L_{\mathcal{I}}=1$.

For any interval $\mathcal{I}'=[L_{\mathcal{I}'},R_{\mathcal{I}'}]$, let $n_{[L_{\mathcal{I}'},R_{\mathcal{I}'}]}$ be the number of new arrivals of $\mathcal{I}'$, and $d_{[L_{\mathcal{I}'},R_{\mathcal{I}'}]}$ be the number of jammed slots within $\mathcal{I}'$. In this proof, define function $f'(x)=f(x)/a$ and $g'(x)=g(x)/a$. For any $t\geq t_0$, we have $n_{[1,t]}\leq t/(64cc_3c_1f'(t))$ and $d_{[1,t]}\leq t/(64cc_3c_1g'(t))$, otherwise the lemma trivially holds by definition of truncated length. Therefore, for any interval $\mathcal{I}'=[L_{\mathcal{I}'},R_{\mathcal{I}'}]\subseteq[1,t]$ whose length $|\mathcal{I}'|=R_{\mathcal{I}'}-L_{\mathcal{I}'}+1$ is at least $t/64cc_3$, we have $n_{[L_{\mathcal{I}'},R_{\mathcal{I}'}]}\leq n_{[1,t]}\leq t/(64cc_3c_1 \cdot f'(t))\leq |\mathcal{I}'|/(c_1 \cdot f'(|\mathcal{I}'|))$ since $|\mathcal{I}'|\leq t \leq 64cc_3|\mathcal{I}'|$; similarly, $d_{[L_{\mathcal{I}'},R_{\mathcal{I}'}]}\leq |\mathcal{I}'|/(c_1 \cdot g'(|\mathcal{I}'|))$.

Call the nodes that start Phase 3 in slot $L_{\mathcal{I}}$ as the \emph{batch nodes} of $\mathcal{I}$. Let $n$ be the number of batch nodes of $\mathcal{I}$. (Throughout this proof, we always use $n$ without subscript to refer to the number of batch nodes of $\mathcal{I}$.)

We first focus on the situation $n=0$. If the complete interval is from a beginning slot to a transition slot, we will apply Lemma \ref{lem: batch without additional nodes} on the channel $\alpha$ determined by the parity of $L_{\mathcal{I}}$ (i.e., the set of slots $\{t\in [L_{\mathcal{I}},R_{\mathcal{I}}]:t-L_{\mathcal{I}}+1 \textrm{ is odd}\}$). We argue the conditions of Lemma \ref{lem: batch without additional nodes} are satisfied. Specifically, setting $k=0$ satisfies the first condition; conditions three and four trivially satisfy. Denote the number of nodes that run \textsc{backoff} on channel $\alpha$ in the first $t/4$ slots of $\alpha$ as $n^{\alpha}_{[1,t/4]}$, then $n^{\alpha}_{[1,t/4]}\leq n_{[1,t/2]}\leq n_{[1,t]}\leq t/(4c_1f'(t))\leq (t/4)/((c_1/a)f(t/4))$. Similarly, the bound on the number of jammed slots is also satisfied. Thus, the second condition of the lemma is satisfied. As a result, apply Lemma \ref{lem: batch without additional nodes} and we know, there is at least one success in the first $t/4$ slots of channel $\alpha$, with high probability in $t/4$. This implies $\bar{l}\leq t/2<t$ with high probability in $t$, as desired. If $n=0$ and the complete interval is from a transition slot to a transition slot or from a transition slot to an ending slot, then there must exist some node(s) running Phase 2 in $L_{\mathcal{I}}$, so these nodes can distinguish between the control channel and the data channel. Again we focus on the first $t/4$ slots of the control channel and apply Lemma \ref{lem: batch without additional nodes} with $k=0$ in the first condition. Therefore, there is at least one success in the first $t/4$ slots of the control channel with high probability in $t/4$, implying $\bar{l}\leq t/2<t$ with high probability in $t$.

Assume $n>0$ in the reminder of the proof. Let $last$ be the slot that the last batch node successfully sends its message among the $n$ batch-nodes. (Let $last$ be $t$ if there are still batch nodes at the end of slot $t$). We consider three scenarios according to the value of $t$: (1) $t>8cc_3n$; (2) $c_4n<t\leq 8cc_3n$; and (3) $t\leq c_4n$. Here, $c_4$ is a constant to be specified later. The analysis for scenario two and three is very similar to the proof of Lemma 8 in \cite{bender20}, we defer them to the appendix to avoid redundancy. Here, we focus on scenario one: $n>0$ and $t>8cc_3n$. We further divide this scenario into four cases.

\textit{Case 1: Suppose $last\geq t/2$.} If $last\geq t/2$, we argue the conditions for applying Lemma \ref{lem: batch with additional nodes} are satisfied, hence there is a success on the control channel during time $[t/4,t/2]$ with high probability in $t/4$, implying $\bar{l}\leq t/2<t$ with high probability in $t$. Specifically, since $n\cdot\frac{c_3\log(t/8)}{t/8}\leq\frac{\log(t/8)}{c'}$ when $t\geq 8c'c_3n$, and since $1\cdot\frac{c_3\log(t/4)}{t/4}\geq\frac{c'\log(t/8)}{t/8}$ when $c_3\geq 2c'$, the first condition is satisfied. Since $n_{[1,t/2]}\leq n_{[1,t]}\leq(t/4)/((c_1/a)f(t/4))$, the bound on the number of injected nodes in the second condition is satisfied. Similarly, the bound on the number of jammed slots in the second condition is also satisfied.

\textit{Case 2: Suppose $last<t/2$, and all nodes never run \textsc{backoff} on control channel within time $[1,last]$.} Then either there is no node in the system at end of slot $last$, or nodes only arrive on data channel after the slot in which the second to last batch node succeeded. In the former case we have $\bar{l}=last<t$. As for the latter case, we apply Lemma \ref{lem: batch without additional nodes} to show there is a success within interval $[last+1,7t/8]$ with high probability in $t$. Specifically, we set $k=0$ for the first condition. The second condition is satisfied since the length of $[last+1,7t/8]$ is at least $3t/8$. The third condition is satisfied since in time slot $last+1$ all nodes in the system begin Phase 2.

\textit{Case 3: Suppose $last<t/2$, $t\geq n^{1.25}$, and there exists some node that runs \textsc{backoff} on control channel within time $[1,last]$.} We first introduce an additional type of ``jammed'' slots, and show the number of such ``jammed'' slots can still be bounded by $\frac{7t/8}{c_1g'(7t/8)}$ for later use. More specifically, call a control channel slot ``interfered'' if any batch node sends in the slot. Define $\sigma=\frac{t}{16c_1g'(t)}$. The number of interfered control slots before time $2\sigma$ can be upper bounded by $\sigma$ trivially. Next we upper bound the number interfered slots after time $2\sigma$. Let $x_i$ be the random number of batch nodes that send in the $i$-th control slot, then for all $i\geq\sigma$, $\mathbb{E}[x_i]\le\frac{c_3\log\sigma}{\sigma}\cdot n\leq t^{-0.9}\cdot n\leq t^{-0.9}\cdot t^{0.8}=t^{-0.1}$, where the last inequality is due to $t\geq n^{5/4}$. Let $y_i$ be an indicator random variable taking value one when at least one batch node sends in the $i$-th slot of the control channel. By Markov's inequality, $\Pr[y_i]=\Pr[x_i\geq 1]\leq\mathbb{E}[x_i]\leq t^{-0.1}$. Hence, the expected number of such interfered control slots within time interval $(2\sigma,t]$ is $\sum_{\sigma<i\le t/2}\Pr[y_i]\leq t^{-0.1}\cdot t/2\leq \frac{t}{20c_1g'(t)}$. By a Chernoff bound, with high probability in $t$, $\sum_{\sigma<i\le t/2}{y_i}$ is at most $\frac{t}{16c_1g'(t)}$. Therefore, if we see interfered slots as another kind of jamming, the total number of jammed slots (i.e., interfered slots and jamming from the adversary) in $[1,t]$ is bounded by $\frac{t}{16c_1g'(t)}+\frac{t}{16c_1g'(t)}+\frac{t}{16c_1g'(t)}\leq \frac{3t/16}{c_1g'(7t/8)}$, where the third term is due to $d_{[1,t]}\leq t/(16c_1g'(t))$.

Denote the first time that some node begins \textsc{backoff} on control channel as $first$ (such slot exists by case assumption), then $first\leq last< t/2$. We now apply Lemma \ref{lem: batch without additional nodes} on the control channel within interval $[first,7t/8]$, with $k=0$ for the first condition, and consider both interfered slots and jamming from adversary for the second condition. We can conclude, with high probability in $t/4$, there is a success (due to nodes running \textsc{backoff}, instead of the batch nodes) on the control channel within interval $[first,7t/8]$. This further implies $\bar{l}\leq 7t/8<t$ holds with high probability in $t$.

\textit{Case 4: Suppose $last<t/2$, $t<n^{1.25}$, and there exists some node that runs \textsc{backoff} on control channel within time $[1,last]$.} We first show with high probability in $n$ (thus also in $t$ as $t<n^{1.25}$), $last>4cc_3n$, by invoking the following claim with $c'$ set to $2cc_3$. (Since we seek a lower bound on the number of slots that all $n$ batch nodes successfully sends their messages on the data channel, without loss of generality, assume there is no new arrival or jammed slots during the first $2cc_3n$ slots of the data channel.)

\begin{claim}\label{claim: batch can not be short}
Assume there are $n$ nodes running $h^{data}$-\textsc{batch} which begins at time 1 on a fixed channel, where $h^{data}(x)=1/x$. Then for any constant $c'$, for sufficiently large $n$, there is at least one node that has not succeeded during the first $c'n$ slots of the channel, with high probability in $n$.
\end{claim}

\begin{proof}
Consider the first time slot $\ell$ when there are only $0.5n$ batch nodes in the system. Let $\ell_i$ be the first time slot that $i$ of those nodes have left the system. Define $\ell_0=\ell$, and let $s_i=\ell_{i+1}-\ell_{i}$. We will prove $\Pr[\sum_{0\le i\le 0.5n-1}s_i\le c'n]<1/n^3$, which implies that with high probability in $n$, the batch runs for at least $c'n$ slots.

Since before time slot $\ell$ there are $0.5n$ times of successes, we have $\ell\ge 0.5n$, which means the probability that each node sends after $\ell$ is at most $\frac{1}{0.5n}$. Thus the probability that some of the $0.5n-i$ nodes send in a slot between $\ell_i$ to $\ell_{i+1}$ is at most $(0.5n-i)\cdot\frac{1}{0.5n}$. Since the occurrence of success implies one of them sends a message, we have $\Pr[s_i\ge L]\ge\left(1-p_i\right)^L$ where $p_i=\frac{0.5n-i}{0.5n}$ for any positive integer $L$. 

Define $\bar{s_i}$ as the truncated geometric distribution where $\Pr[\bar{s_i}=L]=p_i(1-p_i)^L$ for $0\le L<\frac{2}{p_i}$, and $\Pr[\bar{s_i}=\frac{2}{p_i}]=(1-p_i)^{\frac{2}{p_i}}$. It is easy to verify $\Pr[s_i\ge L]\ge\Pr[\bar{s_i}\ge L]$ for any $L\ge 0$. Moreover, we have $\mathbb{E}[\bar{s_i}]\ge\frac{1}{2p_i}$. Write $S=\sum_{0\le i\le 0.5n-1}\bar{s_i}$, we have $\mathbb{E}[S]\ge 0.25n\ln(0.5n)$. By applying the Hoeffding's inequality~\cite{dubhashi12}, we get
$$\Pr[S<c'n]\le 2\exp\left(-\frac{2\left(0.25n\ln(0.5n)-c'n\right)^2}{\sum_{0\le i\le 0.5n-1}\left(\frac{n}{0.5n-i}\right)^2}\right)=\frac{1}{n^{\omega(1)}}$$
\end{proof}

Now assume $last\geq 4cc_3n$. We intend to apply Lemma \ref{lem: batch without additional nodes} on the control channel within time interval $[1,t/2]$. For the first condition, $k=last/2$ is a feasible value, since $n\cdot\frac{c_3\log(last/4)}{last/4}\leq\frac{\log(k/2)}{c'}$ when $last\geq 4c'c_3n$, and $1\cdot\frac{c_3\log(last/2)}{last/2}\geq\frac{c'\log k}{k}$ when $c_3\geq c'$. The second condition is satisfied because the considered interval has length $\Theta(t)$. Recall that there is some node running \textsc{backoff} on control channel within time $[1,last]$ by case assumption, so the third condition of Lemma \ref{lem: batch without additional nodes} is satisfied. Therefore, we conclude with high probability in $t$, there is a success on the control channel before time $t/2$, which implies $\bar{l}\leq t/2<t$.
\end{proof}

\textit{Remark.} Claim \ref{claim: batch can not be short} could be of independent interest, in that it showcases $h^{data}$-\textsc{batch}---a standard implementation of binary exponential backoff---cannot send all $n$ messages in $O(n)$ slots, with high probability in $n$. This holds even if $n$ nodes start simultaneously and there is no external interference.
Nevertheless, assuming all nodes start simultaneously, with high probability in $n$, $h^{data}$-\textsc{batch} can send a constant fraction of all $n$ messages in $O(n)$ slots, even if a constant fraction of all these slots are jammed. (See, e.g., the analysis for scenario two in the appendix.)

\subsection{Proof of the Algorithmic Result}

In this subsection, we will first prove our main algorithmic result, and then prove a corollary clarifying the connection between an algorithm's $(f,g)$-throughput and the number of successful transmission it can guarantee.

\begin{proof}[Proof of Theorem \ref{thm:upper-bound}]
Recall the definition of $(f,g)$-throughput. We focus on the first $t$ slots and assume $n_t$ nodes arrive within these slots. By definition of complete interval, the number of active slots  in $[1,t]$ is the summation of the lengths of the complete intervals that any of these $n_t$ nodes involved in, excluding any active slots after $t$. Denote these complete intervals as $\mathcal{I}_1,\mathcal{I}_2,\cdots,\mathcal{I}_{n_t}$, where some may be empty. (The number of such complete intervals is at most $n_t$ since there is at least one success during each complete interval.) Specifically, if $\mathcal{I}_{last}$ is the last complete interval that begins in some slot in $[1,t]$, then all $\{\mathcal{I}_{k}\}_{k>last}$ has length $0$. Since we focus on the number of active slots in the first $t$ slots, we can assume nodes that are still active at the end of slot $t$ are allowed to continue the algorithm after slot $t$ while Eve does not inject new nodes or jam after slot $t$, and bound the number of active slots in this setting instead.

For each complete interval $\mathcal{I}_j$, we use $m_{\mathcal{I}_j}$ to denote the number of successes occurred during $\mathcal{I}_j$, use $n_{\mathcal{I}_j}$ to denote the number of new arrivals during $\mathcal{I}_j$, and use $d_{\mathcal{I}_j}$ to denote the number of jammed slots during $\mathcal{I}_j$. Recall that $l_{\mathcal{I}_j}$ is the length of $\mathcal{I}_j$, and $\bar{l}_{\mathcal{I}_j}$ is the truncated length of $\mathcal{I}_j$. If $l_{\mathcal{I}_j}>\bar{l}_{\mathcal{I}_j}$ occurs (i.e., $\bar{l}_{\mathcal{I}_j}=0$), then by the definition of truncated length, at least one of following three conditions holds: (1) $m_{\mathcal{I}_j}\geq l_{\mathcal{I}_j}/(32cc_3(t_0+2))$, called the \emph{many-success condition}; (2) $n_{\mathcal{I}_j}> l_{\mathcal{I}_j}/(64cc_3c_1\cdot f'(l_{\mathcal{I}_j}))$, called the \emph{heavy-arriving condition}; (3) $d_{\mathcal{I}_j}>l_{\mathcal{I}_j}/(64cc_3c_1\cdot g'(l_{\mathcal{I}_j}))$, called the \emph{heavy-jamming condition}. (As in the proof of Lemma \ref{lem: bounding length of a complete interval}, define function $f'(x)=f(x)/a$ and $g'(x)=g(x)/a$.) Let $\mathcal{C}_1$ (respectively, $\mathcal{C}_2$ or $\mathcal{C}_3$) be the set containing the complete intervals that satisfy the many-success condition (respectively, heavy-arriving condition or heavy-jamming condition).

\marginnote{Remove usage of $f'$ and $g'$ from the whole paper.}
We intend to bound the total length of all intervals in $\mathcal{C}_1$, $\mathcal{C}_2$, and $\mathcal{C}_3$. But before that, we first show any complete interval that satisfies the heavy-arriving condition or the heavy-jamming condition will have length $O(t)$, otherwise the theorem already holds. Recall that $f$ and $\log g$ are sub-logarithmic, thus for any constant $\hat{c}>0$, there exists some constant $c_0$ such that for any $x\in\mathbb{N}^+$, we have$|f(\hat{c}x)-f(x)|\leq c_0$, implying $f'(\hat{c}x)/f'(x)\leq c_0$; and $|\log g(\hat{c}x)-\log g(x)|\leq c_0$, implying $g'(\hat{c}x)/g'(x)\leq 2^{c_0}$. Let $\gamma_0=2^{c_0}\cdot 64cc_3c_1$. Now, if there exists some $\mathcal{I}_j\in\mathcal{C}_3$ satisfying $l_{\mathcal{I}_j}\geq\gamma t$ for some $\gamma>\gamma_0$, then by the definition of the heavy-jamming condition, $d_{l_{\mathcal{I}_j}}>\gamma t/(64cc_3c_1\cdot g'(\gamma t))\geq 2^{c_0}t/(2^{c_0}\cdot g'(t))\geq t/g'(t)$, which further implies $n_tf(t)+d_tg(t)\geq ad_tg'(t)\geq ad_{l_{\mathcal{I}_j}}\cdot g'(t)> at\geq t$, thus the number of active slots in $[1,t]$ is trivially at most $n_tf(t)+d_tg(t)$, and the theorem is proved. Similarly, if there exists some $\mathcal{I}_j\in\mathcal{C}_2$ satisfying $l_{\mathcal{I}_j}\geq\gamma t$ for some $\gamma>\gamma_0$, then by the definition of the heavy-arriving condition, $n_{l_{\mathcal{I}_j}}= \gamma t/(64cc_3c_1\cdot f'(\gamma t))\geq 2^{c_0}t/(c_0\cdot f'(t))\geq t/f'(t)$, which further implies $n_tf(t)+d_tg(t)>t$, thus the theorem trivially holds. Therefore, from now on, we assume there is no $\mathcal{I}_j$ that satisfies the heavy-arriving condition or the heavy-jamming condition but with length larger than $\gamma t$.

We now bound the total length of all intervals in $\mathcal{C}_1$, $\mathcal{C}_2$, and $\mathcal{C}_3$:
\textit{(1) Complete intervals satisfying many-success condition.} Since $\sum_{\mathcal{I}_j \in \mathcal{C}_1} m_{\mathcal{I}_j}\leq n_t$, we have
$\sum_{\mathcal{I}_j\in\mathcal{C}_1}l_{\mathcal{I}_j}
\leq\sum_{\mathcal{I}_j\in\mathcal{C}_1}m_{\mathcal{I}_j}\cdot 32cc_3(t_0+2)
\leq 32cc_3(t_0+2)\cdot n_t$.
\textit{(2) Complete intervals satisfying heavy-arriving condition.} Since $\sum_{\mathcal{I}_j \in \mathcal{C}_2} n_{\mathcal{I}_j}\leq 2n_t$ (as each node is a new arrival for at most two distinct complete intervals), we have $\sum_{\mathcal{I}_j\in\mathcal{C}_2}l_{\mathcal{I}_j}
\leq\sum_{\mathcal{I}_j\in\mathcal{C}_2}n_{\mathcal{I}_j}\cdot(64cc_3c_1)\cdot f'(l_{\mathcal{I}_j})
\leq\sum_{\mathcal{I}_j\in\mathcal{C}_2}n_{\mathcal{I}_j}\cdot(64cc_3c_1)\cdot f'(\gamma_0 t)
\leq 2^{c_0}\cdot 128cc_3c_1\cdot n_t\cdot f'(t)$.
\textit{(3) Complete intervals satisfying heavy-jamming condition.} Since $\sum_{\mathcal{I}_j \in \mathcal{C}_3} d_{\mathcal{I}_j}\leq d_t$, we have $\sum_{\mathcal{I}_j\in\mathcal{C}_3}l_{\mathcal{I}_j}
\leq\sum_{\mathcal{I}_j\in\mathcal{C}_3}d_{\mathcal{I}_j}\cdot(64cc_3c_1)\cdot g'(\mathcal{I}_j)
\leq\sum_{\mathcal{I}_j\in\mathcal{C}_3}d_{\mathcal{I}_j}\cdot(64cc_3c_1)\cdot g'(\gamma_0 t)
\leq 2^{c_0}\cdot 64cc_3c_1\cdot d_t \cdot g'(t)$.

As the final preparation for bounding total number of active slots, we show $\sum_{j=1}^{n_t} \bar{l}_{\mathcal{I}_j}$ is $O(n_t)$. For any $j\in [n]$, if $\mathcal{I}_j$ is not empty, then by Lemma \ref{lem: bounding length of a complete interval}, conditioned on any history up to the beginning of $\mathcal{I}_j$, we have $\Pr[\bar{l}_{\mathcal{I}_j}=t]\leq t^{-13}$ for $t\geq t_0$; and if $\mathcal{I}_j$ is empty (which implies $\bar{l}_{\mathcal{I}_j}=0$), then again we have $\Pr[\bar{l}_{\mathcal{I}_j}=t]=0\leq t^{-13}$. Apply Lemma \ref{lem: special case of sum of dependent random variables} and we know, with high probability in $n_t$, $\sum_{j=1}^{n_t} \bar{l}_{\mathcal{I}_j}\leq (t_0+2)n_t$.

In conclusion, if we set $a=(t_0+2)+32cc_3(t_0+2)+128cc_3c_1\cdot 2^{c_0}$ to be a sufficiently large constant, then we have
$\sum_{j=1}^{n_t}l_{\mathcal{I}_j}
\leq \sum_{j=1}^{n_t}\bar{l}_{\mathcal{I}_j} + \sum_{\mathcal{I}_j\in\mathcal{C}_1}l_{\mathcal{I}_j} + \sum_{\mathcal{I}_j\in\mathcal{C}_2}l_{\mathcal{I}_j} + \sum_{\mathcal{I}_j\in\mathcal{C}_3}l_{\mathcal{I}_j}
\leq  a\cdot n_t\cdot f'(n_t) + a\cdot d_t\cdot g'(t)$.
Recall $f(x)=af'(x)$, $g(x)=ag'(x)$, so the number of active slots in first $t$ slots in at most $n_t\cdot f(t)+d_t\cdot g(t)$.
\end{proof}


Next, we state and prove the following corollary connecting an algorithm's $(f,g)$-throughput and the number of successful transmission it can guarantee.

\begin{corollary}\label{cor:throughput-and-successes}
Assume nodes run an algorithm that achieves $(f,g)$-throughput. For interval $[1,t]$, an adversary strategy is called "\emph{smooth}" if for any $1\le j\le t-1$, the number of nodes arrived in $[t-j,t]$ is small enough in $O(j/f(j))$ and the number of jammed slots in $[t-j,t]$ is small enough in $O(j/g(j))$. Then under any "smooth" adversary strategy $\mathcal{B}$, for any $1\le j\le t-1$, all nodes arrived before slot $t-j$ will leave the system by the end of slot $t$, with high probability in $j$.
\end{corollary}

\begin{proof}
Fix an arbitrary integer $j\in[1,t-1]$. For any $\mathcal{B}$, construct another adversary strategy $\mathcal{B'}$ in the following way. Besides the same node arrivals and jammed slots as $\mathcal{B}$, $\mathcal{B'}$ injects another $j^{0.5}$ nodes in the last slot (i.e., slot $t$) and jams the last slot. Notice that the execution under adversary strategy $\mathcal{B}$ and $\mathcal{B'}$ are identical until the end of slot $t-1$. Moreover, in $\mathcal{B'}$ nodes cannot succeed in slot $t$ since that slot is jammed. Thus, if we can prove the claim in $\mathcal{B'}$ then the claim also holds in $\mathcal{B}$. Therefore, we focus on $\mathcal{B'}$ in the reminder of the proof.

Suppose to the contrary, there exists some node arrived before slot $t-j$ that leaves the system after slot $t$. Then all slots in $[t-j,t]$ are active. Therefore, there must exist some integer $k\in[0,t-j-1]$ such that all slots in $[k+1,t]$ are active and slot $k$ is inactive. (In case $k=0$ then all slots in $[1,t]$ are active.) Consider interval $[k+1,t]$. Since the adversary strategy $\mathcal{B}$ is smooth, for $\mathcal{B'}$ the number of arrived nodes and jammed slots in $[k+1,t]$ is small enough in $O((t-k)/f(t-k))+j^{0.5}=O((t-k)/f(t-k))$ (recall $f$ is sub-logarithmic and $t-k>j$) and $O((t-k)/g(t-k))+1=O((t-k)/g(t-k))$, respectively. Therefore, according to Theorem~\ref{thm:upper-bound} and the definition of $(f,g)$-throughput, (at least) with high probability in $j^{0.5}$ (thus also in $j$), the number of active slots in $[k+1,t]$ is (strictly) less than $t-k$, a contradiction. Summing over the distribution of $k$ we get with high probability in $j$, all nodes arrived before slot $t-j$ will leave the system by the end of slot $t$.
\end{proof}

\section{Impossibility Results}\label{sec-lower-bound}

In this section, we will prove our impossibility result on the trade-off between the best possible throughput and the severity of jamming, we will also prove another impossibility result demonstrating the necessity of \textsc{backoff} style procedures in attending optimal throughput when jamming is present.

\smallskip\noindent\textbf{Main impossibility result.} The impossibility result on the trade-off exploits the following dilemma that all contention resolution algorithms must confront: on the one hand, when few nodes (e.g., only one node) are in the system, each node's broadcasting probability should be sufficiently high, otherwise successes will not happen fast enough; on the other hand, however, when a lot of nodes are in the system, each node's broadcasting probability should be sufficiently low, otherwise contention among themselves would prevent successes from happening.

The following lemma captures the case that when many nodes are in the system, each node's broadcasting probability cannot be too high for too long.

\begin{lemma}\label{lem:large-contention}
Let $h$ be a sub-logarithmic function and $\mathcal{A}$ be a contention resolution algorithm. Consider a node that runs $\mathcal{A}$. Suppose in expectation, the node broadcasts $\omega\left(h(t)\log t\right)$ times in the first $t$ slots (since its activation) before it hears the first success, then $\mathcal{A}$ does not achieve $(h,g)$-throughput for any $g$.
\end{lemma}

\begin{proof}
Let $x_i$ be the probability that $\mathcal{A}$ sends a messages in the $i$-th slot since the execution of $\mathcal{A}$ starts, assuming no successes occur in slots $1$ to $i-1$.
Consider an adversary strategy that injects $(3\log{t})/x_1$ nodes in each of the first $\sqrt{t}$ slots---call these nodes ``batch-injected'', and injects another $t/(2h(t))$ nodes in the first $t$ slots uniformly
at random---call these nodes ``random-injected''. We show no success occurs in the first $t$ slots, with high probability in $t$.

For each of the first $\sqrt{t}$ slots, since each newly arrived node will send in the slot with probability $x_1$ and there are at least $(3\log{t})/x_1$ newly arrived nodes, the probability of a success occurring is at most $1/t^3$. (This is because, in each such slot, at least $(3\log{t})/x_1-1$ newly arrived batch-injected nodes need to choose not send.)

Consider a slot $k>\sqrt{t}$ and a random-injected node $u$, the probability that $u$ sends in this slot is $\sum_{1\le j\le k}(1/t)\cdot x_{k-j+1}=E_k/t$. (The probability that "$u$ arrives in the $j$-th slot" times the probability that "$u$ sends in the $(k-j+1)$-th slot since its arrival".) Here, $E_k=\sum_{1\le j\le k}x_{k-j+1}=\sum_{1\le j\le k}x_j$. Due to lemma assumption, we know $E_k=\omega(h(k)\cdot\log{k})$. Since $h$ is sub-logarithmic and $k>\sqrt{t}$, we have $h(k)=\Omega(h(t))$, thus $E_k=\omega(h(t)\cdot\log{t})$. As a result, the expected contention in slot $k$ is at least $(t/(2h(t)))\cdot(E_k/t)=\omega(\log t)$, implying the probability of success in slot $k$ is at most $1/t^3$.

Apply a union bound, we know with high probability in $t$ there are no successes in the first $t$ slots.

Assume $t$ is sufficiently large. Since $h$ is sub-logarithmic, $h(t)=O(\log t)$ and is at least $10$. Hence, the number of injected nodes in the first $t$ slots is $(3\log{t})/x_1\cdot\sqrt{t} + t/(2h(t)) \le t/(1.5h(t))$. This means $n_t\le t/(1.5h(t))\leq t$ and $d_t=0$. Now, if algorithm $\mathcal{A}$ achieves $(h,g)$-throughput, then the number of active slots in the first $t$ slots is bounded by $t/(1.5h(t))\cdot h(n_t)\leq t/(1.5h(t))\cdot h(t)=t/1.5$. However, we know all the $t$ slots are active, a contradiction.
\end{proof}

To prove the impossibility result, what remains is to show that in case few nodes are in the system, each node must generate enough contention to ensure a success can happen fast enough.

\begin{proof}[Proof of Theorem \ref{thm:lower-bound}]
Suppose algorithm $\mathcal{A}$ achieves $(f,g)$-throughput and satisfies: $f,\log(g)$ are both sub-logarithmic and $f(x)=o((\log{x})/\log^2g(x))$. We will show, in expectation, $\mathcal{A}$ broadcasts $\Omega((\log^2{t})/\log^2 g(t))$ times in the first $t$ slots
before any success occurs. Together with Lemma \ref{lem:large-contention}, the theorem is immediate.

Consider an interval of $t$ slots. Consider an adversary strategy that injects one node $u$ in the first slot,
\marginnote{Do we need to inject nodes in the last slot of the interval?}
and jams the first $t/(4g(t))$ slots as well as the last slot. The adversary also jams another $t/(4g(t))$ slots which are chosen uniformly at random from slot interval $(t/(4g(t)),t]$. Since there are at most $t/(2g(t))+1$ jammed slots and $t/(3f(t))+1$ injected nodes, and since $\mathcal{A}$ achieves $(f,g)$-throughput, the number of active slots in the first $t$ slots is at most $t-1$ with high probability in $t/f(t)$, for sufficiently large $t$. Therefore, there is at least one success in the first $t$ slots with high probability in $t/f(t)>\sqrt{t}$, as $f(t)=O(\log t)$.

Denote $k(t)$ as the number of times node $u$ broadcasts in slots $(t/(4g(t)),t]$. Since in slot interval $(t/(4g(t)),t]$ the adversary randomly chooses $t/(4g(t))$ slots to jam, the probability that no success occur in these slots is at least $\sum_{k\le\frac{t}{8g(t)}}\frac{\Pr[k(t)=k]}{(8g(t))^k}$. Since there must be a success in these slots with probability at least $1-1/\sqrt{t}$, we have $\sum_{k\le\frac{t}{8g(t)}}\frac{\Pr[k(t)=k]}{(8g(t))^k}\le\frac{1}{\sqrt{t}}$. Since  $t/(8g(t))\ge\log_{8g(t)}\frac{\sqrt{t}}{2}$ for the family of $g$ we care, we have:
$$\sum_{k\le\frac{t}{8g(t)}}\frac{\Pr[k(t)=k]}{(8g(t))^k}\geq
\sum_{k\le\log_{8g(t)}\frac{\sqrt{t}}{2}}\frac{\Pr[k(t)=k]}{(8g(t))^{\log_{8g(t)}\sqrt{t}/2}}=
\sum_{k\le\log_{8g(t)}\frac{\sqrt{t}}{2}}\frac{\Pr[k(t)=k]}{\sqrt{t}/2}=
\frac{\Pr\left[k(t)\le\log_{8g(t)}\left(\sqrt{t}/2\right)\right]}{\sqrt{t}/2}$$
If $\Pr\left[k(t)\le\log_{8g(t)}\frac{\sqrt{t}}{2}\right]>\frac{1}{2}$, then $\sum_{k\le\frac{t}{8g(t)}}\frac{\Pr[k(t)=k]}{(8g(t))^k}>\frac{1/2}{\sqrt{t}/2}=\frac{1}{\sqrt{t}}$, which is a contradiction. As a result, we know $\Pr\left[k(t)\le\log_{8g(t)}\frac{\sqrt{t}}{2}\right]\le\frac{1}{2}$, which means $\Pr\left[k(t)\ge\log_{8g(t)}{\frac{\sqrt{t}}{2}}\right]\ge\frac{1}{2}$, implying $\mathbb{E}[k(t)]\ge\frac{1}{2}\log_{8g(t)}\frac{\sqrt{t}}{2}$.

\marginnote{This paragraph (or this entire proof) needs better presentation.}
By an argument similar as above, for any $i\in[l]$ where $l=\log_{4g(t)}t$, in slot interval $\left(\frac{t}{(4g(t))^{i}}\middle/\left(4g\left(\frac{t}{(4g(t))^{i}}\right)\right),\frac{t}{(4g(t))^i}\right]$, the expected number of times node $u$ broadcasts is at least $\frac{1}{2}\log_{8g(t/(4g(t))^i)}\frac{\sqrt{t/(4g(t))^i}}{2}$.
Since $g$ is non-decreasing, we have $\frac{t}{(4g(t))^{i}}/\left(4g\left(\frac{t}{(4g(t))^{i}}\right)\right)\ge \frac{t}{(4g(t))^{i+1}}$, implying in expectation node $u$ must broadcast $\frac{1}{2}\log_{8g(t/(4g(t))^i)}\frac{\sqrt{t/(4g(t))^i}}{2}$ times during slots $\left(\frac{t}{(4g(t))^{i+1}},\frac{t}{(4g(t))^i}\right]$.
Let $b_t$ be the expected number of times $u$ broadcasts in the first $t$ slots, we have:
$$b_t
\ge\sum_{0\le i\le l/2}\frac{1}{2}\cdot\frac{\log\left(\sqrt{\frac{t}{(4g(t))^i}}\middle/2\right)}{\log\left(8g\left(\frac{t}{(4g(t))^i}\right)\right)}
\ge\frac{l}{4}\cdot\frac{\log\left(\sqrt{\frac{t}{(4g(t))^{l/2}}}\middle/2\right)}{\log\left(8g\left(\frac{t}{(4g(t))^0}\right)\right)}
=\Omega\left(l\cdot\frac{\log t}{\log g(t)}\right)
=\Omega\left(\frac{\log^2t}{\log^2g(t)}\right)$$

Recall the first paragraph of this proof, the theorem is proved.
\end{proof}

\textit{Remark.} Our proof critically relies on the absence of collision detection. To see this, for each node $u$, denote $x^u_i$ as the probability that it broadcasts in the $i$-th slot since its activation, assuming no success occurs in slots 1 to $i-1$. Without collision detection, for any two nodes $v$ and $w$, for any $i\geq 1$, the distribution of $x^v_i$ and $x^w_i$ are identical, as $v$ and $w$ receive identical channel feedback. Thus the superscript is not necessary: we can use $x_i$ to denote $x^u_i$. However, with collision detection, this no longer holds. For example, if $v$ hears silence in the first slot since its activation and $w$ hears collision in the first slot since its activation, then $x^v_2$ and $x^w_2$ might differ! More fundamentally, collision detection breaks the dilemma the impossibility result proof exploits: nodes can \emph{differentiate} whether there are few nodes/low contention (hearing silence), or there are many nodes/high contention (hearing collision), and take different actions.

\smallskip\noindent\textbf{Necessity of \textsc{backoff}.} Recall that in our algorithm, two variants of the standard backoff procedure are used: (1) \textsc{batch}, which sends with probability $1/i$ in slot $i$; and (2) \textsc{backoff}, which chooses several slots in a stage to send and then increases stage length. A critical difference between these two procedures is that \textsc{backoff} is \emph{adaptive}: in a slot, the sending probability of a node depends on the previous sending behavior of the node. The next theorem shows that \textsc{backoff} is necessary when jamming exists: non-adaptive sending pattern cannot provide optimal throughput.

\begin{theorem}\label{thm:non-adaptive-lower-bound}
For algorithm $\mathcal{A}$, if $\mathcal{A}$ will send with pre-defined probability $a_i$ in the $i$-th slot since it starts and before any success is heard, then for any function $f,g$ such that $f,\log(g)$ are both sub-logarithmic and $f(x)=o\left(\frac{\log{x}}{\log g(x)}\right)$, algorithm $\mathcal{A}$ does not achieve $(f,g)$-throughput.
\end{theorem}

\begin{proof}[Proof sketch]
Consider an interval of $t$ slots, we will show $\sum_{1\le i\le t}a_i=\Omega\left(\frac{\log^2t}{\log g(t)}\right)$, and then apply Lemma \ref{lem:large-contention} to obtain our conclusion. The following part is similar to the proof of Theorem~\ref{thm:lower-bound}.

For the sake of contradiction, assume $\mathcal{A}$ achieves $(f,g)$-throughput. Consider the adversary strategy that jams the first $\frac{t}{4g(t)}$ slots as well as the last slot, and injects $2$ nodes in the first slot as well as $\frac{t}{4f(t)}$ nodes in the last slot. We can show the probability that there is no success in the first $t$ slots is at most $1/\sqrt{t}$.

Denote $p_i$ as the probability that the $i$-th slot succeeds. Let interval $I=\left(\frac{t}{4g(t)},t\right]$. For every $i\in I$, we know $p_i=2a_i(1-a_i)\leq 1/2$. Due to the analysis in the last paragraph, we also have $\prod_{i\in I}(1-p_i)\le 1/\sqrt{t}$. Notice that when $0\leq p_i\leq 1/2$, we have $\prod_{i\in I}(1-p_i)\geq\prod_{i\in I}4^{-p_i}=4^{-\sum_{i\in I}p_i}$, hence $4^{-\sum_{i\in I}p_i}\leq 1/\sqrt{t}$, implying $\sum_{i\in I}p_i=\Omega(\log{t})$. Since $a_i\ge p_i/2$, we have $\sum_{i\in I}a_i=\Omega(\log t)$. Recall the proof of Theorem \ref{thm:lower-bound}, consider intervals $\left(\frac{t}{\left(4g(t)\right)^{i+1}},\frac{t}{\left(4g(t)\right)^i}\right]$, we have $\sum_{1\leq i\leq t}a_i=\Omega\left(\frac{\log^2 t}{\log g(t)}\right)$. Recall the first paragraph of this proof, the theorem is proved.
\end{proof}

\begin{acks}
The authors would like to thank Prof.\ Yitong Yin for the comments and suggestions that greatly improve the overall quality of the paper.
\end{acks}

\bibliographystyle{ACM-Reference-Format}
\bibliography{./podc-draft-3}

\appendix
\section*{Appendix}

\begin{proof}[\underline{Proof of Lemma \ref{lem: special case of sum of dependent random variables}}]
For any $i\in[n]$, for any values $x_1,\cdots,x_{i-1}$ of $X_1,\cdots,X_{i-1}$, we have $\mathbb{E}[X_i~|~X_1=x_1,\cdots,X_{i-1}=x_{i-1}]\leq t_0\cdot\Pr[X_i\leq t_0~|~X_1=x_1,\cdots,X_{i-1}=x_{i-1}]+\sum_{t>t_0}t\cdot\Pr[X_i=t~|~X_1=x_1,\cdots,X_{i-1}=x_{i-1}]\leq t_0+\sum_{t>t_0}t\cdot t^{-13}\leq t_0+1$. Besides, $\Pr[X_i\geq n^{0.2}]\leq\sum_{t\geq n^{0.2}}t^{-13}\leq n^{-2}/12$, for sufficiently large $n$. Take a union bound over the $n$ random variables, we know with high probability in $n$, $X_i\leq n^{0.2}$ holds for each $i\in[n]$. Assume indeed $X_i$ is at most $n^{0.2}$ for each $i\in[n]$. Apply Lemma 3 from \cite{bender20}, we can show with high probability in $n$, $\sum_{i=1}^{n} X_i\leq (t_0+2)n$.
\end{proof}

\begin{proof}[\underline{Missing parts in proof of Lemma \ref{lem: bounding length of a complete interval}}]
Here we provide analysis for scenario two and scenario three, which is very similar to the proof of Lemma 8 in \cite{bender20}.

\textit{Scenario II: $2c_4n<t\leq 8cc_3n$.} Let $\eta=\frac{c_4}{4(t_0+2)}$. We will prove with high probability in $n$ (thus also in $t$), there are at least $\eta n$ successes on the data channel within interval $[1,2c_4n]$. If this holds, in the case that the \textsc{batch} ends in the first $2c_4n$ slots (i.e. $R_{\mathcal{I}}\leq L_{\mathcal{I}}+2c_4n-1$), we have $\bar{l}\leq 2c_4n <t$;\footnote{In such case, it must be that the \textsc{batch} ends due to some success on the control channel, since by Claim~\ref{claim: batch can not be short} it take $\omega(n)$ slots to generate $n$ successes on the data channel. So the end of \textsc{batch} also means the end of the complete interval the lemma is considering.} in the other case that the \textsc{batch} lasts for at least $2c_4n$ slots, we have $\bar{l}=0$ since there are $\eta n\geq t/(32cc_3(t_0+2))$ successes, where the inequality is due to $t<8cc_3n$.

To prove there are at least $\eta n$ successes on the data channel within $[1,2c_4n]$, we begin with some notations. Let $\tau=c_4n$. Let $s_i$ be the number of slots (of the data channel) from $L_{\mathcal{I}}$ to the slot that the $i$-th success of the data channel occurs, and define $s_0=0$.
Further we set $s_i$ as $\tau$ if $s_i$ exceeds $\tau$.
Let $X_i$ be the interval from $(s_{i-1}+1)$-th slot to $s_{i}$-th slot of the data channel, thus $s_i=\sum_{j=1}^{i}|X_j|$. So what we need to prove is $s_{\eta n}\leq \tau$. We also use $n_{X_i}$ and $d_{X_i}$ to denote the number of new arrivals and jammed slots in interval $X_i$, respectively. Recall we assume $n_{[1,2\tau]}\leq n_{[1,t]}\leq t/(64cc_3c_1f'(t))\leq \tau/(8c_1f'(\tau))$ and $d_{[1,2\tau]}\leq d_{[1,t]}\leq d_{\mathcal{I}}\leq t/(64cc_3c_1g'(t))\leq \tau/(8c_1g'(\tau))$.

We define three types of intervals, and for each type bound the total length of the intervals in $\{X_i\}$ belonging to that type. We call an interval $X$ \emph{heavy-arriving} if $n_{X}>|X|/(c_1f'(|X|))$; call $X$ \emph{heavy-jamming} if $d_{X}>|X|/(c_1g'(|X|))$; and call $X$ \emph{light} if $X$ is not heavy-arriving and not heavy-jamming. Let $\mathcal{X}_{HeavyA}$ (respectively, $\mathcal{X}_{HeavyJ}$ or $\mathcal{X}_{Light}$) be the set containing all heavy-arriving intervals (respectively, all heavy-jamming intervals or all light intervals).  (These three sets are not necessarily disjoint.) We know $\sum_{X_i\in \mathcal{X}_{HeavyA}}|X_i|\leq \sum_{X_i\in \mathcal{X}_{HeavyA}}n_{X_i}\cdot c_1f'(|X_i|)\leq \frac{\tau}{8c_1f'(\tau)}\cdot c_1f'(|X_i|)\leq \tau/8$ since $\sum_{X_i\in \mathcal{X}_{HeavyA}} n_{X_i}\leq n_{[1,2\tau]}\leq \tau/(8c_1f'(\tau))$. Similarly we can prove $\sum_{X_i\in \mathcal{X}_{HeavyJ}}|X_i|\leq \tau/8$.

What remains it to bound the total length of light intervals. We further divide light intervals into three groups and bound the total length of intervals within each group. For light intervals that end by slot $n$ (of the data channel), we can bound them trivially since $\sum_{X_i\in \mathcal{X}_{Light} \wedge s_{i}\leq n}|X_i|\leq n\leq \tau/4$ when $c_4\geq 4$.

The second group contains at most one interval $X_i$ that begins before or at slot $n$ (which implies last success happened before slot $n$), and ends after slot $n$. We can show with high probability in $n$ this interval ends by $n+\tau/8\leq \tau/4$ when $c_4\geq 8$, conditioned on any fixed value of $i$ and any fixed values of $s_1,s_2,\cdots,s_{i-1}$, by the following claim.

\begin{claim*}
Fix any $i\leq\eta n$, and any values of $s_1,s_2,\cdots,s_{i-1}$ satisfying $s_{i-1}<n$. Then with high probability in $n$, there is at least a success in $(n,n+\tau/8]$ (thus also in interval $(n,\tau/4]$ when $c_4\geq 8$).
\end{claim*}

\begin{proof}
Notice interval $[1,n+\tau/8]$ is light, since $n_{[1,2\tau]}\leq \frac{\tau/8}{c_1f'(\tau/8)}$ and $d_{[1,2\tau]}\leq \frac{\tau/8}{c_1g'(\tau/8)}$. We intend to apply Lemma \ref{lem: batch with additional nodes} on interval $[n+1,n+\tau/8]$ and argue its conditions are satisfied. When $n$ is sufficiently large with respect to $c$ and $c_4$, we have $n\cdot \frac{1}{n}\leq\frac{\log(\tau/4)}{c'}$ and $(1-\eta)n\cdot\frac{1}{n+\tau/8}\geq\frac{0.9n}{\tau/4}\geq\frac{c'\log(\tau/4)}{\tau/4}$, thus the first condition for applying the lemma is satisfied. Moreover, since interval $[1,n+\tau/8]$ is light, it is easy to verify the second condition for applying the lemma is also satisfied. Therefore, by Lemma \ref{lem: batch with additional nodes}, with high probability in $n$, there is a success in $[n+1,n+\tau/8]$.
\end{proof}

The third group of light intervals consists of the intervals that begin after slot $n$. We need the following claim and Lemma \ref{lem: special case of sum of dependent random variables} to show $\sum_{X_i\in \mathcal{X}_{Light} \wedge s_{i-1}> n}|X_i|\leq \eta \cdot n \cdot (t_0+2) \leq c_4n/4=\tau/4$.

\begin{claim*}
Fix any $i\leq\eta n$, and any values of $s_1,s_2,\cdots,s_{i-1}$ satisfying $s_{i-1}\geq n$. Recall each $s_j$ is at most $c_4n$. For any $x\geq t_0$, with high probability in $x$, $\Pr(|X_i|>x\wedge X_i\in\mathcal{X}_{Light})\leq 1/x^{\Omega(1)}$, where $t_0$ is a sufficiently large constant.
\end{claim*}

\begin{proof}
Let $x'\geq x$ be the smallest integer such that $(s_{i-1},s_{i-1}+x']$ is light. It suffices to upper bound the probability
$$\Pr\left((s_{i-1},s_{i-1}+x']\textrm{ is light }\wedge(s_{i-1},s_{i-1}+x')\textrm{ contains no success }\right)$$
We intend to apply Lemma \ref{lem: batch with additional nodes} on interval $(s_{i-1},s_{i-1}+x']$ and argue its conditions are satisfied. Specifically, since $(s_{i-1},s_{i-1}+x']$ is light, the second condition is satisfied. On the other hand, if we set $t_0\geq 2^{c'}$, then $n\cdot\frac{1}{n}\leq\frac{\log x'}{c'}$ since $x'\geq t_0$. Moreover, if we set $t_0\geq(2c'c_4)^2$, then $0.9n\cdot\frac{1}{c_4n}\geq\frac{0.9}{c_4}\geq\frac{c'\log x'}{x'}$ since $n-\eta n\geq 0.9n$ and $x'\geq t_0$. Hence, the first condition is also satisfied. Now, apply Lemma \ref{lem: batch with additional nodes} on interval $(s_{i-1},s_{i-1}+x']$, we conclude there is a success within interval $(s_{i-1},s_{i-1}+x']$ with high probability in $x'$ (thus also in $x$).
\end{proof}

We can now conclude the proof for scenario two. Specifically, $s_{\eta n}\leq\tau/8+\tau/8+\tau/4+\tau/4+\tau/4\leq\tau$ by summing the total length of each type of interval. Therefore, there are at least $\eta n$ successes in the first $\tau$ slots of the data channel, with high probability in $n$ (thus also in $t$), implying $\bar{l}=0$.

\textit{Scenario III: $t\leq 2c_4n$.} When $c_3\geq 12c_4$, we have $2c_4n\leq c_3n/6$, so assume $t\leq c_3n/6$. In the case that there are at least $2n/3=\frac{2nc_4}{3c_4}\geq\frac{t}{3c_4}\geq t/(32cc_3(t_0+2))$ successes within time $[1,t]$, we have $\bar{l}=0<t$. Otherwise, there are at least $n/3$ batch nodes within time $[1,t]$, and we argue there is no success in the first $t/2$ slots on the control channel, with high probability in $n$. This is because in each of the $t/2$ slots, the contention of batch nodes is at least $n/3\cdot\frac{c_3\log(c_3n/6)}{c_3n/6}\geq 2\log(c_3n/6)\in\Omega(\log{n})$. Therefore, in this case, with high probability in $n$ (thus also in $t$), $R_{\mathcal{I}}-L_{\mathcal{I}}+1>t$ (since there is no success on control channel and there are active nodes in the system). In conclusion, for this scenario, with high probability in $t$, the value of $\bar{l}$ is either more than $t$ or $0$.
\end{proof}

\end{document}